\documentclass[11pt,reqno]{amsart}

\usepackage{color}
\usepackage[colorlinks=true, allcolors=blue]{hyperref}
\usepackage{amsmath, amssymb, amsthm}
\usepackage{mathrsfs}
\usepackage{mathtools}
\usepackage[noabbrev,capitalize,nameinlink]{cleveref}
\crefname{equation}{}{}
\usepackage{fullpage}
\usepackage[noadjust]{cite}
\usepackage{graphics}
\usepackage{pifont}
\usepackage{tikz}
\usepackage{bbm}
\usepackage[T1]{fontenc}

\usetikzlibrary{arrows.meta}
\usetikzlibrary{decorations.pathreplacing}

\usepackage{environ}
\usepackage{framed}
\usepackage{url}
\usepackage[linesnumbered,ruled,vlined]{algorithm2e}
\usepackage[noend]{algpseudocode}
\usepackage[labelfont=bf]{caption}
\usepackage{cite}
\usepackage{framed}
\usepackage[framemethod=tikz]{mdframed}
\usepackage{appendix}
\usepackage{graphicx}
\usepackage[textsize=tiny]{todonotes}
\usepackage{tcolorbox}
\usepackage{enumerate}
\allowdisplaybreaks[1]
\usepackage{enumerate}

\crefname{algocf}{Algorithm}{Algorithms}

\crefname{equation}{}{}
\AtBeginEnvironment{appendices}{\crefalias{section}{appendix}}

\usepackage[color,final]{showkeys}

\colorlet{refkey}{orange!20}
\colorlet{labelkey}{blue!30}

\crefname{algocf}{Algorithm}{Algorithms}

\numberwithin{equation}{section}
\newtheorem{theorem}{Theorem}[section]

\newtheorem{lemma}[theorem]{Lemma}

\crefname{claim}{Claim}{Claims}

\newtheorem*{question*}{Question}

\theoremstyle{definition}
\newtheorem{definition}[theorem]{Definition}

\newtheorem*{definition*}{Definition}

\theoremstyle{remark}

\newcommand{\mb}{\mathbb}

\newcommand{\mbm}{\mathbbm}
\newcommand{\mc}{\mathcal}

\newcommand{\eps}{\epsilon}
\newcommand{\assign}{\leftarrow}
\newcommand{\tsc}{\textsc}

\allowdisplaybreaks

\title{Online Edge Coloring via Tree Recurrences and Correlation Decay}

\author[A1]{Janardhan Kulkarni}
\address{Microsoft Research, Redmond.}
\email{jakul@microsoft.com.}

\author[A2]{Yang P. Liu}
\address{Department of Mathematics, Stanford University,
Stanford, CA 94305, USA}
\email{yangpliu@stanford.edu}

\author[A3]{Ashwin Sah}
\author[A4]{Mehtaab Sawhney}
\address{Department of Mathematics, Massachusetts Institute of Technology, Cambridge, MA 02139, USA}
\email{\{asah,msawhney\}@mit.edu}

\author[A5]{Jakub Tarnawski}
\address{Microsoft Research, Redmond.}
\email{jatarnaw@microsoft.com.}

\begin{document}

\begin{abstract}
We give an online algorithm that with high probability computes a $\left(\frac{e}{e-1} + o(1)\right)\Delta$ edge coloring on a graph $G$ with maximum degree $\Delta = \omega(\log n)$ under online {\em edge arrivals} against oblivious adversaries, making first progress on the conjecture of Bar-Noy, Motwani, and Naor in this general setting. Our algorithm is based on reducing to a matching problem on locally treelike graphs, and then applying a tree recurrences based approach for arguing correlation decay. 
\end{abstract}

\maketitle

\section{Introduction}
Given a graph $G:=(V, E)$ with maximum degree $\Delta$, the edge coloring problem is to assign colors to edges such that any two edges sharing a common vertex get different colors.
A well-known theorem by Vizing \cite{vizing64} says that every graph can be edge-colored using $\Delta + 1$ colors; furthermore, such a coloring can be found in polynomial time.
From an algorithmic standpoint, a remarkable aspect of Vizing's theorem is that it achieves the optimal bound for the problem, as $\Delta$ colors are necessary for every graph with maximum degree $\Delta$ and it is NP-hard to distinguish whether a graph needs $\Delta$ or $\Delta+1$ colors \cite{holyer81,kochol10}.

Given a proper edge coloring of a graph, each color class induces a matching; thus, any edge coloring of a graph partitions the edge set into a collection of matchings.
This view of edge coloring plays an important role in its applications to routing in switching networks and reconfigurable topologies \cite{BMN92,aggarwal2003switch}.
In routing applications, however, the edges of graph arrive {\em online}, which models the arrival of new traffic that needs to be routed between two switches or servers.
This was the motivation that led Bar-Noy, Motwani, and Naor \cite{BMN92} to initiate the study of edge coloring in the online setting.
Here, the online algorithm has knowledge of the vertex set $V$ of the graph and the maximum degree $\Delta$. 
However, the edges are revealed one by one, and the online algorithm has to irrevocably assign a color to each newly arriving edge. 
The goal is to minimize the number of colors used by the online algorithm while maintaining a valid edge coloring of the graph at all time steps.
In the original paper, \cite{BMN92} showed that for graphs with maximum degree $O(\log n)$, no online algorithm can maintain a proper coloring using fewer than $2\Delta-1$ colors
--- a trivial bound achieved by the greedy algorithm which simply assigns every arriving edge any color that is not used at either endpoint.
However, this result should be interpreted as a lower bound on the additive error rather than multiplicative, as it only applies to graphs with logarithmic maximum degree.
Consequently, the focus has shifted to the much more interesting regime of $\Delta = \omega(\log n)$.
In this regime, \cite{BMN92} conjectured that the online algorithm that uses $\Delta + O(\sqrt \Delta \log n)$ colors and samples a color for each edge uniformly at random from the set of valid colors succeeds with constant probability.
However, we do not know how to analyze this algorithm or give any online algorithm that beats the competitive ratio of 2 achieved by the trivial greedy algorithm; this has been raised as a challenging open problem by all subsequent works.

\medskip

Over the past three decades, significant attempts have been made towards resolving this conjecture. 
A competitive ratio of $1+o(1)$ is achievable in important special cases: random-order (instead of adversarial-order) edge arrival \cite{aggarwal2003switch,bahmani2012online,BGW21} and one-sided vertex arrivals (instead of edge arrivals) for bipartite graphs \cite{CohenPW19}. Very recently, a competitive ratio of 1.9 was obtained by Saberi and Wajc~\cite{SW21} for general \emph{vertex} arrivals on graphs of maximum degree $\omega(\log n)$.
Despite these impressive results, which we discuss in detail in \cref{sec:related}, no algorithm was known to beat the competitive ratio of 2 in the most general setting of online edge arrivals considered in the conjecture of \cite{BMN92}.
Moreover, the barrier of 2 for the edge coloring problem seemed to parallel a similar barrier for the ``dual''
problem of online matching.
Namely, a surprising recent result~\cite{gamlath2019online} showed a lower bound of 2 for online matching in the edge arrival setting (whereas better algorithms exist for the vertex arrival setting~\cite{karp1990optimal,DJK13,gamlath2019online}), and it was conceivable that the online edge coloring problem might exhibit the same dichotomy~\cite{BGW21} (although in the fractional case, edge coloring is trivial while matching is not).
The main result of this paper shows a separation between these two problems, and makes the first progress towards resolving the conjecture of Bar-Noy, Motwani, and Naor.

\begin{theorem}\label{thm:main}
There is an online randomized algorithm that on a graph $G$ with maximum degree $\Delta = \omega(\log n)$ outputs an $\left(\frac{e}{e-1} + O\left(\frac{(\log \log \Delta)^2}{\log \Delta}+(\log n/\Delta)^{1/4}\right)\right)\Delta$-edge coloring with high probability in the oblivious adversary setting.
\end{theorem}

Our proof of the theorem is based on reducing the problem to a matching problem on locally treelike graphs, and then applying a {\em tree recurrences} based approach for arguing {\em correlation decay}.
We believe that both our algorithm and its analysis are quite simple.
Correlation decay is a well known and widely used technique in the statistical physics and sampling literature \cite{Wei06, sly2008uniqueness, bandyopadhyay2008counting, li2013correlation, vigoda2000improved, anari2021spectral, srivastava2014counting}, but has not been applied for online problems.
Our work shows the efficacy of this technique in the analysis of online algorithms,  and we believe that it has potential for broader applicability in settings that need to cope with input uncertainty such as online, recourse, dynamic, or streaming problems.

\medskip

Before we proceed, we make some remarks regarding our algorithm. First, while na\"{i}vely implementing the algorithm as given in this paper yields a running time of $\tilde{O}(\Delta)$ per edge, using basic data structures and standard efficient sampling primitives allows one to implement the algorithm in time $\tilde{O}(1)$ per edge. Additionally, straightforward modifications of our algorithm show that our theorem extends for multigraphs with maximum edge multiplicity bounded by $o(\Delta)$. Finally, our techniques give an alternate proof (given in \cref{sec:randomorder}) that one can $(1+o(1))\Delta$-edge color a graph with maximum degree $\Delta$ if its edges arrive in a uniformly random order, which was the main result of \cite{BGW21}:

\begin{theorem}\label{thm:random}
There is an online randomized algorithm that on a graph $G$ with maximum degree $\Delta = \omega(\log n)$ and a uniformly random ordering of edge arrivals outputs a $(\Delta+o(\Delta))$-edge coloring with high probability.
\end{theorem}

The $o(\Delta)$ term can be quantified to be $O(\Delta (\Delta')^{-1/24}+\sqrt{\Delta\Delta'\log n})$, where we set 
\[ \Delta' = \min(\log \Delta/\log\log\Delta,\sqrt{\Delta/\log n}). \] While this is quantitatively worse than the $O((\Delta/\log n)^{-c})$ dependence in \cite{BGW21}, we believe that our new proof further demonstrates the utility of our subsampling and tree recurrences based approach. More importantly, it reveals a deeper structural reason why a $(1+o(1))\Delta$-edge coloring is easier to achieve in the random-order model -- the edges which determine whether an edge $e$ is matched in a subsampled graph (under random-order edge arrivals) form a tree. This tree case is substantially simpler, and should become clear to the reader in our overview \cref{sec:techniques}. 

\subsection*{Outline}
We give an intuitive version of our entire argument arc in \cref{sec:techniques}.
We present the history of the problem in \cref{sec:related}.
\cref{sec:reduction,sec:ouralgo} comprise the formal proof of our result  (\cref{thm:main}).
We conclude and discuss future directions in \cref{sec:conclusions}.
\cref{sec:randomorder} contains the proof for the random-order case (\cref{thm:random}).

\section{Our Techniques and Proof Overview}
\label{sec:techniques}
We start by introducing some background and then give an overview of our proof. 
This section is not a prerequisite for the full proof (\cref{sec:reduction,sec:ouralgo}),
and
readers who prefer to see our techniques in full detail can skip it;
indeed, our proof is quite short.

Recall that in any edge coloring,
the set of edges of a single color forms a matching.
In the reverse direction, a natural reduction due to Cohen, Peng, and Wajc~\cite{CohenPW19}
shows that an $(\alpha + o(1))\Delta$-coloring can be achieved
by repeatedly invoking a matching algorithm that matches each edge with probability at least $1/(\alpha \Delta)$
(and assigning a new color to the matching edges, then removing them from the graph).
A vertex of degree $\Delta$ will be matched with probability roughly $1/\alpha$, so the maximum degree decreases at a rate of roughly $1$ per $\alpha$ iterations.
Having $\Delta = \omega(\log n)$ yields enough concentration for this process to finish in $(\alpha + o(1))\Delta$ iterations. This is the only point where our algorithm (and the previous works \cite{CohenPW19,SW21}) uses $\Delta = \omega(\log n)$. Saberi and Wajc~\cite{SW21} observed that this reduction works for any arrival model.
Therefore we are left with the task of designing an algorithm
that matches every edge with probability at least $(1-o(1))/(\frac{e}{e-1}\Delta)$, with edges arriving online in any order against an oblivious adversary.

Online algorithms that match each edge with probability at least $1/(\alpha \Delta)$ for $\alpha < 2$ are known for
bipartite graphs under {\em vertex} arrivals; see \cite{cohen2018randomized} for an example.
The positive results~\cite{CohenPW19, SW21} for the online edge coloring problem build upon and extend these results.
Furthermore, these algorithms require many new ideas over \cite{cohen2018randomized}, including a novel LP relaxation and sophisticated online rounding schemes.
Unfortunately, it appears that these techniques rely critically on the vertex arrival model and do not generalize to edge arrivals.
Indeed, Saberi and Wajc \cite{SW21} write that the edge arrival setting "remains out of reach'' with the known techniques.
The main technical contribution of this paper is a simple argument based on correlation decay to overcome the limitations of the previous works.

\subsection*{Reduction to locally treelike instances by subsampling.}
To solve the online matching problem,
we begin by subsampling each edge of the graph with probability roughly $\Delta' / \Delta$ for some $\Delta' = \omega(1)$.
This will decrease the maximum degree to roughly $\Delta'$, and we in fact trim any edges above that degree threshold to ensure this.
More crucially for our arguments, the subsampled graph will be \emph{locally treelike}:
for most edges $e$, their close neighborhood will contain no cycles, and hence form a tree.

To see this, let us note that for any $\ell$,
the number of length-$\ell$ cycles containing $e$ \emph{in the original graph}
is at most $\Delta^{\ell - 2}$ (start walking from one endpoint of $e$ and make $\ell-1$ choices of neighbor, where the last one is forced to end up at the other endpoint of $e$).
However, conditioned on $e$ being in the sparsified graph, each of them survives in the sparsified graph with probability only at most $(\Delta' / \Delta)^{\ell-1}$.
Thus, the expected number of cycles of length up to some threshold $g$ is at most
$\sum_{\ell = 3}^g (\Delta')^{\ell-1} / \Delta \le  (\Delta')^g / \Delta$,
which is $o(1)$ if $(\Delta')^g = o(\log n)$ -- the reader should imagine that $\Delta', g = \omega(1)$ are arbitrarily slowly growing functions of $n$.
We can thus imagine that we delete any edge that would cause a short cycle to appear.\footnote{
We could afford to do so in the algorithm,
but for simplicity we instead argue in the analysis that for any edge $e$, there is only a small probability that there is a short cycle in the $g$-neighborhood of $e$ (not necessarily containing $e$).
}

Now, we want to match every surviving edge with probability at least $1/C$ for $C \approx \frac{e}{e-1} \Delta'$.
Then the probability that an edge survives and gets matched is at least
$(1-o(1))(\Delta'/\Delta) \cdot 1 / (\frac{e}{e-1} \Delta') = (1-o(1)) / (\frac{e}{e-1} \Delta)$, as desired by the edge-coloring-to-matching reduction.

\subsection*{Online matching on trees.}
Let us first see how to get a good algorithm for online matching in an ideal scenario where the (subsampled) graph is in fact a tree.
We will guarantee that every edge is matched with probability $1/C$.
We now describe what to do when an edge $(u, v)$ arrives.
Of course, we can only match it if $u$ and $v$ are not yet matched.
Let $d_u$ denote the number of edges already adjacent on $u$ (not counting $e$).
The probability that any of these edges is matched is $1/C$;
as these are disjoint events, the probability that $u$ is not yet matched is $1-\frac{d_u}{C}$.
Similarly, for $v$ it is $1-\frac{d_v}{C}$ where $d_v$ is $v$'s degree.
Crucially, as the graph is a tree,
these events are independent
(as $u$ and $v$ are in different connected components before the arrival of $e$).
Thus, if we match $e$ with probability $\frac{C}{(C-d_u)(C-d_v)}$ in the case $u$ and $v$ are unmatched,\footnote{This value is in $[0,1]$ as long as $C \ge \Delta' + \Omega(\sqrt{\Delta'})$.}
then we get the required overall probability of
\begin{equation}
\label{eq:matching_probability}
\left(1-\frac{d_v}{C} \right) \left(1-\frac{d_u}{C} \right) \frac{C}{(C-d_u)(C-d_v)} = \frac{1}{C} .
\end{equation}
Hence this algorithm inductively matches each edge with probability exactly $1/C$, as desired.

We remark that Cohen and Wajc~\cite{cohen2018randomized},
who give a $(1+o(1))$-competitive algorithm
for online matching in  regular graphs
{under one-sided bipartite vertex arrivals},
similarly sample edges $(u,v)$ adjacent to an arriving vertex $v$
with probability proportional to $\frac{C}{C-d_u}$,
so as to get marginal probability $\approx 1/C$ for each edge,
though their algorithm is more complex.

\subsection*{Online matching on locally treelike graphs.}
When the graph is not a tree,
the difficulty is that the above events of $u$ and $v$ being unmatched
are not independent.
Nevertheless, we argue that in the absence of short cycles,
these events are not very correlated.
In fact, our algorithm is the same as in the tree case.

\begin{figure}[t]
    \centering

    \begin{tikzpicture}[scale=1.1]
    
    \tikzstyle{vertex}=[circle, draw,fill=gray!60, inner sep=0pt, minimum size=12pt]
    \tikzstyle{ignoredvertex}=[circle, draw,fill=red!10, inner sep=0pt, minimum size=12pt]
    \newcommand{\myheight}{1.5}
    \newcommand{\normalthick}{1.5pt}
    \newcommand{\morethick}{3.5pt}
    
    \node[vertex] (v1) at (-2.5,3.5) {$u$};
    \node[vertex] (v2) at (0,3.5) {$v$};
    \draw  (v1) edge[line width=\normalthick] node[below] {21} node[above] {\large $e$} (v2);
    \node[vertex] (v4) at (-4,3.5-\myheight) {};
    \node[ignoredvertex] (v3) at (-3,3.5-\myheight) {};
    \node[vertex] (v8) at (0.5,3.5-\myheight) {};
    \node[vertex] (v9) at (1.5,3.5-\myheight) {};
    \node[ignoredvertex] (v7) at (-3,3.5-2*\myheight) {};
    \node[vertex] (v6) at (-4,3.5-2*\myheight) {};
    \node[vertex] (v5) at (-5,3.5-2*\myheight) {};
    \node[ignoredvertex] (v12) at (0,3.5-2*\myheight) {};
    \node[vertex] (v10) at (1,3.5-2*\myheight) {};
    \node[vertex] (v11) at (2,3.5-2*\myheight) {};
    \draw  (v1) edge[dashed,red,line width=\normalthick] node[right] {22} (v3);
    \draw  (v1) edge[line width=\normalthick] node[left] {19} (v4);
    \draw  (v4) edge[line width=\normalthick] node[left] {14} (v5);
    \draw  (v6) edge[line width=\normalthick] node[right] {12} (v4);
    \draw  (v3) edge[dotted,red,line width=\normalthick] node[right] {3} (v7);
    \draw  (v2) edge[line width=\normalthick] node[left] {11} (v8);
    \draw  (v2) edge[line width=\normalthick] node[right] {20} (v9);
    \draw  (v9) edge[line width=\normalthick] node[left] {17} (v10);
    \draw  (v9) edge[line width=\normalthick] node[right] {10} (v11);
    \draw  (v8) edge[dashed,red,line width=\normalthick] node[left] {15} (v12);
    \node[vertex] (v21) at (-4,3.5-3*\myheight) {};
    \node[vertex] (v18) at (-5.5,3.5-3*\myheight) {};
    \node[vertex] (v20) at (-6,3.5-4*\myheight) {};
    \node[vertex] (v19) at (-5,3.5-4*\myheight) {};
    \node[vertex] (v22) at (-4,3.5-4*\myheight) {};
    \node[ignoredvertex] (v23) at (-3,3.5-4*\myheight) {};
    \node[vertex] (v15) at (1,3.5-3*\myheight) {};
    \node[vertex] (v16) at (1,3.5-4*\myheight) {};
    \node[ignoredvertex] (v13) at (0,3.5-3*\myheight) {};
    \node[ignoredvertex] (v14) at (-1,3.5-3*\myheight) {};
    \node[ignoredvertex] (v17) at (0,3.5-4*\myheight) {};
    \draw  (v12) edge[dashed,red,line width=\normalthick] node[right] {18} (v13);
    \draw  (v12) edge[dotted,red,line width=\normalthick] node[left] {1} (v14);
    \draw  (v10) edge[line width=\normalthick] node[right] {8} (v15);
    \draw  (v15) edge[blue,line width=\morethick] node[right] {4} (v16);
    \draw  (v13) edge[dotted,red,line width=\normalthick] node[left] {6} (v17);
    \draw  (v5) edge[line width=\normalthick] node[left] {13} (v18);
    \draw  (v18) edge[line width=\morethick,blue] node[right] {9} (v19);
    \draw  (v18) edge[line width=\morethick,blue] node[right] {2} (v20);
    \draw  (v21) edge[line width=\normalthick] node[right] {7} (v6);
    \draw  (v21) edge[line width=\morethick,blue] node[left] {5} (v22);
    \draw  (v21) edge[dashed,red,line width=\normalthick] node[right] {16} (v23);
    \draw[fill=gray, fill opacity=0.15]  plot[smooth cycle, tension=.7] coordinates {(-3.9,3.3) (-5.8,0.7) (-6.5,-2.5) (-3.8,-3.1) (-3,2.8) (0.4,1.6) (0.6,-3) (2.2,-1) (1.9,2.6)  (-0.5,4)};
    \node at (-5.3,2.7) {$W(T)$};
    \draw [decorate,decoration={brace,amplitude=10pt,mirror,raise=4pt}] (2.4,-2.7) -- (2.4,-0.9) node [midway,xshift=0.8cm] {$\partial T$};
    \end{tikzpicture}
    \caption{
    An edge $e = (u,v)$ together with its neighborhood, which contains no cycles.
    The numbers on edges specify the order of arrival (with 1 coming first).
    The red (dashed/dotted) edges can be ignored, as they are not part of the witness tree $W(T)$ (\cref{def:witnesstree}): dashed edges arrive after their parent-edge, and dotted edges belong to subtrees of dashed edges.
    The blue (thick) edges are boundary edges, which the adversary controls.
    We show that the adversary can decide them deterministically upfront, and moreover this decision should be to leave them unmatched.
    Thus we reduce our analysis to the case where only the black (non-boundary) edges arrive, and are matched randomly as in our \cref{algo:matching}.
    Note that then, the events of matching $u$ and $v$ (before $e$ arrives) are independent.
    }
    \label{fig:tree}
\end{figure}
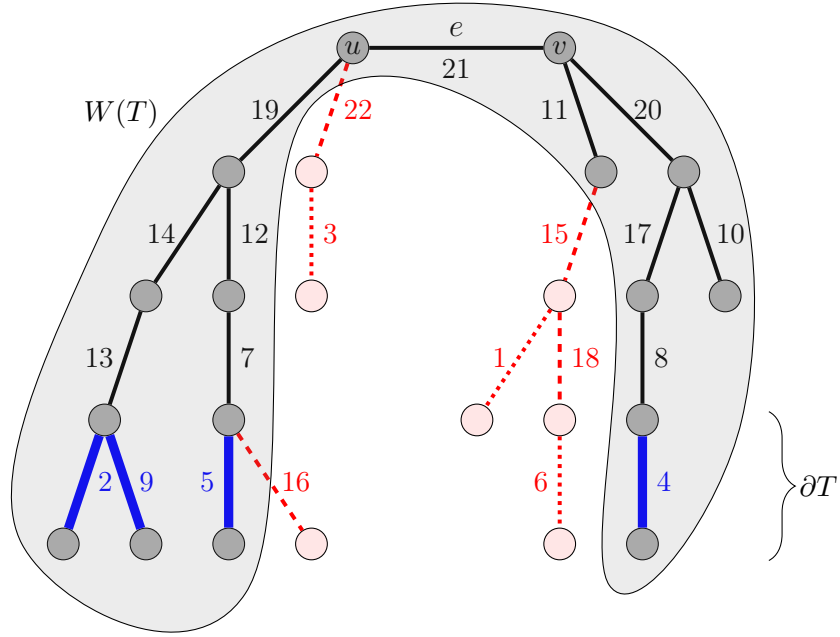

Assume that edge $e=(u,v)$ has no cycles in its neighborhood of radius $g$,
which we denote by $T$, as it is a tree (see \cref{fig:tree} for an example).
Intuitively, any correlation between the above two events is due to some $u$-$v$-path that must pass over the boundary of $T$;
we will show \emph{correlation decay} as we go up the tall-enough tree.

\subsection*{Edge matching game.}
We lower-bound the probability that $e$ is matched
by considering a worst-case scenario
where we cede control over all the boundary edges of $T$
to an adversary,
whose objective is to minimize the probability that $e$ is matched.
In this \emph{edge matching game} played on $T$,
the powerful adversary is allowed to match or not match any arriving boundary edge of $T$; his decisions may depend on the partial matching built up to that point and may be randomized.
All non-boundary edges are matched randomly as in our algorithm.

We may assume without loss of generality that
every edge in $T$ arrives before its parent edge.
Otherwise we can ignore this edge (together with its entire subtree -- see \cref{fig:tree}),
as it only influences the probability of matching $e$ via the boundary,
which the adversary already controls anyway.
In other words, we are in a setting where edges arrive bottom-to-top
(though some boundary edges may arrive after non-boundary edges that are not their ancestors).

\subsection*{Monotonicity.}
What should the adversary do?
Intuitively, to minimize the probability of matching $e$,
the adversary should maximize the probability of matching edges at distance $1$ from $e$. To that end, they should minimize distance-$2$ edges, maximize distance-$3$ edges,
and so on.
We set $g$ to be odd;
then intuitively, the adversary should leave the boundary (distance-$(g+1)$) edges unmatched.
Indeed, we can show a \emph{monotonicity} property:
the lower the probability with which they match a boundary edge,
the lower the probability of matching $e$, even under adaptive decisions.
In particular, this implies that decisions for all boundary edges
can be deterministically fixed at the beginning.
Thus the adversary is effectively eliminated,
and the events of $u$ and $v$ being unmatched (before $e$ arrives) are again independent as in the tree case.

This use of monotonicity to reduce to the ``all unmatched'' case is closely related to reductions done by Weitz \cite{Wei06} when establishing strong spatial mixing for sampling independent sets in the hardcore model in graphs of maximum degree $\Delta$.

\subsection*{Tree recurrences and error decay.}
In our setting, the matching probabilities for vertices at the bottom of $T$ are out of our control.
However, we show that as we go up the tree,
they quickly contract towards our desired value of $1/C$-per-edge.

Define $q_w$ to be the probability that a vertex $w \in T$ is not matched from below.
This happens if every child-edge $(w,w_i)$ of $w$ is not matched;
that is, each $w_i$ was matched from below or the edge $(w,w_i)$ was not sampled.
If $(w,w_i)$ are ordered by time of arrival, we have\footnote{For simplicity, in this introduction we imagine that every non-leaf in $T$ has $\Delta'$ children.}
\begin{align*}
q_w &= \prod_{i=1}^{\Delta'} \mb{P}\left[ \text{$(w,w_i)$ not matched} \mid \text{$(w,w_1), ..., (w,w_{i-1})$ not matched} \right] \\
&= \prod_{i=1}^{\Delta'} \left( 1 - q_{w_i} \cdot \frac{C}{(C-i+1)(C-\Delta')} \right) \,.
\end{align*}
Ideally we would like $q_w$ to be $1 - \frac{\Delta'}{C} = \frac{C-\Delta'}{C}$
(disjoint events of the $\Delta'$ children being matched, each with probability $1/C$),
so let us define the \emph{error} as $\eps_w := 1 - \frac{C}{C-\Delta'} \cdot q_w$.
Some rewriting (see \cref{sec:computing} for more details) then yields the recurrence on error terms:
\begin{align}
    \eps_w &= 1 - \prod_{i=1}^{\Delta'} \left( 1 + \frac{1}{C-i} \cdot \eps_{w_i} \right) \,. \label{eq:upwards}
\end{align}
\newcommand{\epsold}{\ensuremath{\eps_{\mathrm{old}}}}
Imagine that all errors on the level below $w$ are equal: $\epsold := \eps_{w_1} = ... = \eps_{w_{\Delta'}}$. Using $\exp(x) \approx 1+x$ we get:
\begin{align*}
    \eps_w &\approx 1 - \exp \left( \sum_{i=1}^{\Delta'} \frac{1}{C-i} \cdot \epsold \right) \approx 1 - \exp \left( \log\left(\frac{C}{C-\Delta'}\right) \cdot \epsold \right) \approx - \log\left(\frac{C}{C-\Delta'}\right) \cdot \epsold \,.
\end{align*}
Essentially, as we move one level up $T$, the sign of the error flips, and its absolute value is multiplied by $\log\left(\frac{C}{C-\Delta'}\right)$. This multiplier is smaller than $1$ if $C > \frac{e}{e-1} \cdot \Delta'$ -- and this is what gives rise to our competitive ratio.
Then, the errors shrink towards $0$ as we move up the tree $T$.
If the height $g$ of $T$ is made large enough, then
$q_u$ and $q_v$ are very close to $1-\frac{\Delta'}{C}$;
recall that the corresponding events are now independent
(as we have made the boundary deterministic)
and so the probability of matching $e$ is close to $1/C$, as in \eqref{eq:matching_probability}.

This analysis of ``tree recurrences'' is ubiquitous in statistical physics, and our threshold of $C > \frac{e}{e-1}\Delta'$ can be recast in this language as the threshold for having ``uniqueness of the Gibbs measure on the $\Delta'$-ary tree''. More concretely, for $C<\frac{e}{e-1}\Delta'$, the function $f(x) = 1-\exp(x\log\frac{C}{C-\Delta'})$ has nonzero fixed points of order $2$, i.e. $f(f(x)) = x$ for some nonzero $x$, and hence there exist probabilities $p_1 \ne p_2$ such that all edges on even levels are matched with probability $p_1$ and edges on odd levels are matched with probability $p_2$, and yet these probabilities satisfy the necessary recurrence equations. This ``alternate'' fixed state is precisely responsible for the failure of our analysis for $C<\frac{e}{e-1}\Delta'$. 

\section{A Brief History of the Problem}
\label{sec:related}
Prior to our work, there were no known online algorithms for the edge coloring problem in the adversarial edge arrival model, besides the trivial $2\Delta -1$ bound obtained by the greedy algorithm \cite{BMN92}.
All the positive results for the problem are in the setting $\Delta = \omega(\log n)$ and fall into two categories.

\subsection*{Random Arrival of Edges}
Aggarwal, Motwani, Shah, and Zhu \cite{aggarwal2003switch} were the first to show that for very dense multi-graphs with $\Delta = \omega(n^2)$ one can get a near-optimal $(1 + o(1))\Delta$-edge coloring algorithm.
For simple graphs, Bahmani, Mehta, and Motwani~\cite{bahmani2012online} gave a $1.26\Delta$-edge coloring algorithm when $\Delta = \omega(\log n)$. 
Very recently, Bhattacharya, Grandoni, and Wajc~\cite{BGW21} showed that one can get the best of both these results by presenting a $(1 + o(1))\Delta$-edge coloring algorithm for simple graphs with $\Delta = \omega(\log n)$ using an adaptation of the Nibble method.

\subsection*{Adversarial Vertex Arrival Model}
In this model, instead of edges being revealed one by one, vertices of the graph are revealed one at a time, together with adjacent edges to previously revealed vertices.
Cohen, Peng, and Wajc~\cite{CohenPW19} designed an asymptotically optimal $(1 + o(1))\Delta$-edge coloring algorithm for bipartite graphs under one-sided vertex arrivals; that is, the left side of the bipartite graph is fixed and the right vertices arrive in an online fashion.
For general vertex arrivals, Saberi and Wajc \cite{SW21} very recently designed a $(1.9 + o(1))$-competitive randomized algorithm.
This result also applies for general graphs, as there is an online reduction from general graphs to bipartite graphs that works against oblivious adversaries.

\medskip
Cohen, Peng, and Wajc~\cite{CohenPW19} also studied edge coloring when the maximum degree $\Delta$ of the graph is not known. 
For this problem, they showed a lower bound of $e/(e-1)$ for bipartite graphs, even in the setting of one-sided vertex arrivals.
This is in contrast to the known-$\Delta$ case, where the only known lower bound, even for deterministic algorithms under edge arrivals, is an additive error of $\Omega(\sqrt \Delta)$, which follows from a direct reduction of the lower bound for the matching problem in \cite{cohen2018randomized}.

Finally, similarly to \cite{CohenPW19, SW21}, our approach to the online edge coloring problem is via a reduction to the online matching problem, which has been studied extensively for several decades in various arrival models.
We refer the readers to \cite{karp1990optimal, mehta2007adwords, mehta2013online, gamlath2019online, fahrbach2020edge} for an introduction to this literature.

\section{Reduction to Matching on Locally treelike Graphs}
\label{sec:reduction}

Our arguments begin with a standard reduction of online edge coloring to rounding fractional matchings. This is essentially the statement of \cite[Lemma 2.2]{SW21} and exactly the same proof applies.
\begin{lemma}
\label{lemma:matching_reduction}
Let $\mc{A}$ be an online matching algorithm and $\alpha \ge 1$ a parameter such that, given any graph of maximum degree $\widehat{\Delta} = \Omega(\log n)$ as a sequence of edge arrivals, $\mc{A}$ matches each edge with probability at least $1/(\alpha\widehat{\Delta})$. Then there exists an online edge coloring algorithm $\mc{A}'$ that on any graph $G$ with maximum degree $\Delta = \omega(\log n)$ arriving online outputs an $\left(\alpha + O((\log n/\Delta)^{1/4})\right) \Delta$ edge coloring with high probability.
\end{lemma}
We remark that $\mc{A}$ and $\mc{A'}$ need advance knowledge of $\widehat{\Delta}$ and $\Delta$, respectively.

In this work, we apply a further reduction to the online matching problem so that it suffices to consider graphs which are \emph{locally treelike}. More precisely, let the $g$-neighborhood of an edge $e$ in a graph $G'$ denote the set of vertices within distance $g$ of either endpoint. We give a reduction from the graph $G$ to graphs $G'$ which have maximum degree at most $\Delta'$ and we ensure that the $g$-neighborhoods of almost all edges $e \in G'$ have no cycles.
The reader should think of $\Delta' = \omega(1)$, an arbitrarily slowly growing function.

\begin{algorithm}[!ht]
\caption{Computes a subgraph of $G$ online with each edge included with approximately the same probability and with few short cycles. \label{algo:subsample}}
\SetKwProg{Proc}{procedure}{}{}
\Proc{$\textsc{Subsample}(G, \Delta',\Delta)$}{
    $G' \assign \emptyset.$ \Comment{Initially empty subgraph $G'$.} \\
    $d_v \assign 0$ for $v \in V(G).$ \Comment{Number of adjacent sampled edges to vertex $v$} \\
    $\eta = 3\sqrt{(\log \Delta')/\Delta'}$ \\
    \For{$i = 1, \dots, m$}{
        \tcp{Edge $e_i = (u, v)$ arrives}
        $R \assign \tsc{unif}([0,1])$ \\
        \If{$R\le (1-\eta)\Delta'/\Delta$ \label{line:subsample}}{
            \If{$d_u < \Delta'$ \emph{and} $d_v < \Delta'$ \label{line:restrictdegree}}{
                $E(G') \assign E(G') \cup \{ e \}$ \\
            }
            $d_u \assign d_u + 1, d_v \assign d_v + 1.$
        }
    }
}
\end{algorithm}

\begin{lemma}[Uniformly-sampled graphs are locally treelike]
\label{lemma:highgirth}
Let $G$ be a graph with maximum degree $\Delta$. Let $G'$ be a subgraph of $G$ where each edge is included with probability $p = D/\Delta$, and $D \ge 2$. The probability that the $g$-neighborhood of an edge $e \in G'$ contains a cycle is at most $3D^{5g}/\Delta$.
\end{lemma}

\begin{proof}
We first consider the case where the $g$-neighborhood of $e$ has a cycle containing $e$. In this case, there exists a cycle in the $g$-neighborhood with length at most $2g$.

Note that the maximum degree of $G$ is $\Delta$ and hence edge $e \in G$ is in at most $\Delta^{\ell-2}$ cycles of length $\ell.$ Thus the probability that every edge (excluding $e$) in such a cycle is in $G'$ is $(D/\Delta)^{\ell-1}.$ Thus the probability that $e$ is in such a cycle is at most
\[ \sum_{\ell=3}^{2g} \Delta^{\ell-2} \cdot (D/\Delta)^{\ell-1} \le \sum_{\ell=3}^{2g} D^{\ell-1}/\Delta \le D^{2g}/\Delta.\]
We now turn to the case where the cycle $C$ does not contain $e = (u, v)$. Then there must be a path $P$ of length $\ell_P \le g$ from either $u$ or $v$ to some vertex $w$ on $C$, and such that $P$ and $C$ have disjoint edges. Also, $C$ has some length $\ell_C \le 2g.$ Because the maximum degree of $G$ is $\Delta$, the number of pairs of paths and cycles $(P, C)$ satisfying these conditions in $G$ is bounded by $2\Delta^{\ell_P + \ell_C - 1}.$ Hence the probability that some pair $(P, C)$ has all edges in $G'$ is at most
\begin{align*}
~&\sum_{\ell_P=0}^g \sum_{\ell_C=3}^{2g} 2\Delta^{\ell_P + \ell_C - 1} \cdot (D/\Delta)^{\ell_P+\ell_C} \\
=~&2\Delta^{-1} \sum_{\ell_P=0}^g \sum_{\ell_C=3}^{2g} D^{\ell_P + \ell_C} = 2\Delta^{-1}\left(\sum_{\ell_P=0}^g D^{\ell_P} \right)\left(\sum_{\ell_C=3}^{2g} D^{\ell_C} \right) \le 2D^{3g+2}/\Delta \le 2D^{5g}/\Delta.
\end{align*}
The claim follows by combining this with previous case where $e$ is in the cycle.
\end{proof}

\begin{lemma}\label{lemma:degree_reduction}
Given a maximum degree $\Delta$ graph $G$
\cref{algo:subsample} returns a subgraph $G'$ satisfying:
\begin{itemize}
    \item The maximum degree of $G'$ satisfies $\Delta(G') \le \Delta'.$
    \item Any edge $e \in G$ is in $G'$ with probability between \[ \left(1 - 5\sqrt{\frac{\log \Delta'}{\Delta'}}\right)\frac{\Delta'}{\Delta} \quad \qquad \text{and} \quad \qquad \frac{\Delta'}{\Delta} . \]
    \item An edge $e\in G$ is in $G'$ and its $g$-neighborhood contains a cycle with probability at most
    \[\left(3\frac{(\Delta')^{5g}}{\Delta}\right)\frac{\Delta'}{\Delta}. \]
\end{itemize}
\end{lemma}
\begin{proof}
The maximum degree condition in the first bullet point follows immediately by line \ref{line:restrictdegree} of \cref{algo:subsample}. The third bullet point follows immediately from \cref{lemma:highgirth}.

It suffices to show the second bullet point. First note that the probability a given edge is included is always bounded by $(1-\eta)\Delta'/\Delta\le \Delta'/\Delta$ due to the initial sub-sampling. We next lower bound the probability that any vertex $v$ ever reaches the $\Delta'$ threshold. Let $d(v)$ be the degree of $v$ in $G$, and let $X_1, \dots, X_{d(v)}$ denote the events that the edges out of $v$ are initially subsampled in line \ref{line:subsample} of \cref{algo:subsample}. Note that $\mb{E}[X_i] = (1-\eta)\Delta'/\Delta.$ Hence by a Chernoff bound, we have that
\[ \mb{P}\left[\sum_{i=1}^{d(v)} X_i \ge \Delta' \right] \le \exp(-\eta^2\Delta'/3). \]
Thus the probability that an edge $e$ is subsampled by line \ref{line:subsample} of \cref{algo:subsample} and neither of its endpoints $u, v$ violates the degree bound of $\Delta'$ is at least \[ (1-\eta)\frac{\Delta'}{\Delta} \cdot (1-2\exp(-\eta^2\Delta'/3)) \ge (1-\eta-2\exp(-\eta^2\Delta'/3))\frac{\Delta'}{\Delta}. \]
Using that $\eta = 3\sqrt{(\log \Delta')/\Delta'}$ this completes the proof. 
\end{proof}

\section{Algorithmic Description and Tree Recurrences}
\label{sec:ouralgo}

\begin{algorithm}[!ht]
\caption{Computes a matching on a graph $G$ that arrives online with edges $e_1, e_2, \dots, e_m$, with sampling parameter $C > \Delta(G) + 2\sqrt{\Delta(G)}.$ \label{algo:matching}}
\SetKwProg{Proc}{procedure}{}{}
\Proc{$\textsc{Matching}(G, C)$}{
    $\mc{M}\assign \{\}$\Comment{Current set of matched edges} \\
    $d_v \assign 0$ for $v \in V(G).$ \Comment{Current degree of vertex $v$} \\
    $m_v \assign \textbf{False}$ for $v \in V(G).$ \Comment{Records whether vertex $v$ is currently matched} \\
    \For{$i = 1, \dots, m$}{
        \tcp{Edge $e_i = (u, v)$ arrives}
        \If{$m_u = $ \emph{\textbf{False}} \emph{and} $m_v = $ \emph{\textbf{False}}}{
            $R \assign \tsc{unif}([0,1])$ \\
            
            \If{$R\le C/((C-d_u)(C-d_v))$ \label{line:checkr}}{
                $\mc{M}\assign \mc{M} \cup e$ \\
                $m_u \assign \textbf{True}, m_v \assign \textbf{True}.$
            }
        }
        $d_u \assign d_u + 1, d_v \assign d_v + 1.$
    }
}
\end{algorithm}

We will apply \cref{algo:matching} on the subsampled graph.
We now study its effect on treelike graphs.

\begin{theorem}[Matchings on Locally treelike Graphs]
\label{thm:matchingtreelike}
Fix $\delta \in (0, 1/20)$ and assume $\Delta\ge 25$. Consider a graph $G$ with maximum degree $\Delta$ and an edge $e$ whose $g$-neighborhood contains no cycles in $G$. If $C > \left(\frac{e}{e-1} + \delta \right)\Delta$, then $\textsc{Matching}(G, C)$ (\cref{algo:matching}) includes edge $e$ in the final matching with probability at least $\frac{1}{C}\left(1 - (1-\delta/4)^{(g-1)/2} - \frac{10^4}{\delta C}\right)^2$.
\end{theorem}
We will apply \Cref{thm:matchingtreelike} for the graph $G'$ constructed by \cref{lemma:degree_reduction} (whose maximum degree is denoted there by $\Delta'$). To interpret \cref{thm:matchingtreelike}, note that the $g$-neighborhood of $e$ is a tree. Intuitively, our result says that if the number of colors $C$ exceeds a ``critical threshold'' $\frac{e}{e-1}\Delta$ then the correlations between colors on the boundary of the $g$-neighborhood tree of $e$ decay towards the top, and hence $e$ is almost uniform.

To formalize this intuition, we now reduce the study of locally treelike graphs to trees with a given fixed boundary. The analysis is modeled after that developed by Weitz \cite{Wei06} for proving correlation decay in the hardcore model. In particular, we show a key monotonicity claim on subtrees which reduces the analysis of our algorithm to a tree recurrences computation. This observation is closely related to that made by Weitz \cite{Wei06} in studying the hardcore model.

In this case, for an edge $e$ we consider the length $g$ neighborhood of $e$, and denote the graph as $T$ as it is a tree. Define $\partial T$ to be the \emph{boundary edges} of this set, so $\partial T := E(V(T), V(G) \backslash V(T))$. We will imagine that the endpoint of each boundary edge that lies outside of $V(T)$ are distinct -- we will not need to consider collisions between these endpoints. This way, we imagine that $T \cup \partial T$ is a tree. See \cref{fig:tree} for an example.

At a high level we will argue that if the edges in $\partial T$ are chosen to be matched or unmatched adversarially (even adaptively), in the worst case the probability that edge $e$ is matched is still $(1-o(1))/\Delta'$.

Our starting direction is to reduce the potentially complicated arrival ordering of edges in the neighborhood of an edge $e$ to the more natural ordering from bottom-to-top in the tree. To see this we first define the \emph{witness tree} of an edge $e$, which captures all edges in the neighborhood which can possibly influence the probability that $e$ is matched.
\begin{definition}[Witness tree]
\label{def:witnesstree}
We say that an edge $f \in T\cup\partial T$ is \emph{alive} if it is processed before any edge $f'$ that lies strictly above $f$. We let the \emph{witness tree} $W(T)$ be the set of alive edges connected to $e$ in $T\cup\partial T$.
\end{definition}
Note that by definition the edges in the witness tree are processed from leaves (or boundary edges) upwards. Note that there may be edges in the witness tree that are connected to $e$ but not downwards to the boundary.
See \cref{fig:tree} for an example.

\begin{definition}[Matched edge or vertex]
\label{def:matched}
We say that an edge $f$ is matched (at some stage) if it is part of the matching $\mathcal{M}$ at that stage during an invocation of \textsc{Matching}~(\cref{algo:matching}). We say that a vertex $v$ is matched if it has an adjacent matched edge.
\end{definition}

Next we define a game on the witness tree. The goal of the game will be to minimize 
the probability that the top edge $e$ is matched.
\begin{definition}[Edge Matching Game]
\label{def:game}
For an edge $e$ and an ordering of the edges $f_1, \dots, f_{|E(W(T))|}$ of the witness tree $W(T)$ (\cref{def:witnesstree}),
consider the following game. At stage $i$, edge $f_i$ is revealed. If $f_i$ is a boundary edge, i.e. $f_i \in \partial T,$ then the player may either choose to match $f_i$ or not match $f_i$ arbitrarily (possibly in a randomized manner).\footnote{We assume that the player can match a boundary edge $f_i \in \partial T$ even if an adjacent vertex is matched -- this does not affect the proofs later.}
Otherwise, $f_i = (u, v)$ is added to the matching with probability $\frac{C}{(C-d_u)(C-d_v)}$ (as given in \cref{algo:matching}) if neither $u$ or $v$ is matched.
\end{definition}
We can assume that in the edge matching game the edge $e$ is the one processed last. We argue that the probability that an edge $e$ is matched during an invocation of \cref{algo:matching} is lower bounded by the minimum probability that edge $e$ is matched during the edge matching game (\cref{def:game}).

\begin{lemma}[Reduction to Matching Game]
\label{lemma:gamereduction}
Let $m_e$ denote the probability that an edge $e$ is matched during \cref{algo:matching}. Let $m^{\min}_e$ denote the minimum probability that edge $e$ is matched under optimal play in the edge matching game (\cref{def:game}). Then $m^{\min}_e \le m_e$.
\end{lemma}
\begin{proof}
It suffices to show that \cref{algo:matching} can be simulated by the edge matching game for an edge $e$. This is obvious by taking the probability of matching the boundary edge to be the probability that it is matched conditional on all decisions made up to that point.
We also note that edges that do not belong to the witness tree may only influence the probability of $e$ being matched by \cref{algo:matching} indirectly through the boundary edges, but in the edge matching game the adversary has full control over the boundary edges anyway. 
\end{proof}

Next we argue that (somewhat surprisingly) to minimize the probability that edge $e$ is matched in the edge matching game, the player can make choices for boundary edges that are oblivious to any previous choices, i.e. they can all be decided before the start of the game.
To do this it is useful to write out the crucial tree recurrences for calculating the probabilities that a vertex $v$ is matched to a vertex (or unmatched) below it.
\begin{definition}[Tree Recurrence Probabilities]
\label{def:recurrenceformula}
Consider a subgame of the edge matching game where all boundary edges are decided as to whether they are matched and the remaining edges to process form a tree $T'$ that is processed from the leaves upwards. For a vertex $v \in T$ define $q_v$ as the probability that vertex $v$ is \emph{not} matched via some edge below it in the tree on the edge matching game restricted to $T'$.
\end{definition}
Let the current number of children of a vertex $v$ (in $T \backslash T'$) be $c_v'$ and the total number of children by $c_v$. For $i = c'_v+1, \dots, c_v$ let $v_i$ denote the $i$-th vertex under $v$ to be processed during the edge matching game on $T'$.
Clearly, if $v$ is a boundary vertex/edge which is matched by the adversary, then $q_v = 0$. Otherwise $v$ is unmatched if and only if all edges under it are unmatched. Because the events that these edges $(v, v_i)$ want to be matched are independent, i.e. $v_i$ is unmatched and $R \le \frac{C}{(C-d_v)(C-d_{v_i})}$ in line \ref{line:checkr} of \cref{algo:matching}, we get
\begin{align}
    q_v = \prod_{i = c'_v + 1}^{c_v} \left(1 - \frac{C}{(C-i+1)(C-c_{v_i})}q_{v_i} \right).
    \label{eq:recurrence2}
\end{align}
We are now in position to prove the main monotonicity claim.
\begin{lemma}[Oblivious Choices for Boundary Edges]
\label{lemma:obliviouschoice}
If $g$ is odd, then the minimum probability that edge $e$ is matched in the edge matching game is given by the strategy where all boundary edges are always unmatched.
\end{lemma}
\begin{proof}
Consider the edge matching game and let $f$ denote the last boundary edge to be processed. Say that given the game up to this point the player decides to match $f$ with probability $r$. We will use \eqref{eq:recurrence2} to argue that the probability that $e$ is matched is a linear function of $r$ and that the sign depends only on the distance of the boundary edge to $e$. Therefore optimally the player must choose $r = 0$ or $1$ and crucially the choice is independent of all randomness in the game up to this point. Hence the game is equivalent to the situation where edge $f$ is revealed deterministically at the start of time. By inducting on the remaining subgame we get that all edges are revealed deterministically at the start as desired.

It suffices now to verify using \eqref{eq:recurrence2} that the probability that $e$ is matched is a linear function of $r$ and that the sign depends only on the distance of the boundary edge to $e$. To see this note that the sign of $q_{v_i}$ (for fixed variables $q_{v_j} \in [0, 1]$ for $j \neq i$) in \eqref{eq:recurrence2} is $\frac{C}{(C-i+1)(C-c_{v_i})} \in [0, 1].$ Hence the sign of $r$ in the formula for $q_v$ for a vertex $v$ flips every level up the tree. This implies the claim.
\end{proof}

\subsection{Computing the Monotone Tree Matching Game} \label{sec:computing}
By \cref{lemma:obliviouschoice} we know that the boundary edges can be decided deterministically at the start. In this way we can simplify the tree recurrence formula (\cref{def:recurrenceformula}) in \eqref{eq:recurrence2} to the full witness tree $W(T)$
instead of a subtree $T'$. Here recall that $q_v$ is the probability that vertex $v$ is not matched to a vertex below it, and $c_v$ denotes the number of children of $v$ in the tree.
Also note that for any vertices $u, v$ such that neither is a descendent in the tree of another that the events of whether $u$ is matched and $v$ is matched to something below are independent.
\begin{align}
    q_v = \prod_{i=1}^{c_v} \left(1 - \frac{C}{(C-i+1)(C-c_{v_i})}q_{v_i}\right). \label{eq:recurrence}
\end{align}
Let $\Delta := \Delta(G)$. Our goal is to show that for $C > \left(\frac{e}{e-1} + o(1)\right)\Delta$ that the probabilities $q_v$ as we go up the tree contract towards $1 - \frac{c_v}{C} = \frac{C-c_v}{C}$. Thus it is natural to define the errors $\eps_v$ as
$\eps_v = 1 - q_v \cdot \frac{C}{C-c_v}$.
Plugging this into \eqref{eq:recurrence} gives that
\begin{align*}
    q_v &= \prod_{i=1}^{c_v} \left(1 - \frac{C}{(C-i+1)(C-c_{v_i})}q_{v_i} \right) \\
    &= \prod_{i=1}^{c_v} \left(1 - \frac{C}{(C-i+1)(C-c_{v_i})} \cdot \frac{C-c_{v_i}}{C}(1-\eps_{v_i})\right) \\
    &= \prod_{i=1}^{c_v} \left(\frac{C-i}{C-i+1} + \frac{1}{C-i+1}\eps_{v_i} \right) \\
    &= \frac{C-c_v}{C}\prod_{i=1}^{c_v} \left(1 + \frac{1}{C-i}\eps_{v_i}\right).
\end{align*}
Rearranging this equation along with
$q_v \cdot \frac{C}{C-c_v} = 1 - \eps_v$
gives us the recurrence on the error terms
\begin{align}
    \eps_v = 1 - \prod_{i=1}^{c_v}\left(1 + \frac{1}{C-i}\eps_{v_i}\right) \label{eq:epsrecurrence}.
\end{align}

We now start showing that the $\eps_v$ contract as we move up the tree. To this end we define $\eps^{\min}_\ell$ and $\eps^{\max}_\ell$ to the minimum/maximum values of $\eps_v$ for vertices $v$ that lie distance $g-\ell$ from edge $e$.
\begin{definition}
\label{def:minmaxeps}
For an edge $e$ with witness tree $T$ define $\eps^{\min}_\ell$ to the minimum values of $\min(\eps_v, 0)$ for vertices $v$ that lie distance $\ell$ from the boundary of $T$. Define $\eps^{\max}_\ell$ to the maximum values of $\max(\eps_v, 0)$ for vertices $v$ that lie distance $\ell$ from the boundary of $T$.
\end{definition}
Note that for $C \ge \frac{e}{e-1}\Delta$ we have $0 \ge \eps^{\min}_\ell \ge \frac{-c_v}{C-c_v} \ge -2$ and $0 \le \eps^{\max}_\ell \le 1.$

Recall from the proof of \cref{lemma:obliviouschoice} that the signs flip per level up the tree. This allows us to bound $\eps^{\min}_{\ell+1}$ in terms of $\eps^{\max}_\ell$ and similarly $\eps^{\max}_{\ell+1}$ in terms of $\eps^{\min}_\ell$.
\begin{lemma}[Single level error bounds]
\label{lemma:onelevel}
Let $C \ge \frac{e}{e-1}\Delta$ and $\Delta\ge 25$. For any level $\ell \in [0, g)$ we have that
\[ \eps^{\min}_{\ell+1} \ge 1 - \left(1 + \frac{10^2}{C}\right)\exp\left(\log\left(\frac{C}{C-\Delta}\right)\eps^{\max}_\ell \right)\] and
\[
\eps^{\max}_{\ell+1} \le 1 - \left(1 - \frac{10^2}{C}\right)\exp\left(\log\left(\frac{C}{C-\Delta}\right)\eps^{\min}_\ell \right).
\]
\end{lemma}
\begin{proof}
To start we note by integration that
\begin{align}
    \left|\sum_{i=1}^\Delta \frac{1}{C-i} - \log\left(\frac{C}{C-\Delta}\right) \right| \le 5/C.
    \label{eq:riemann}
\end{align}
Now to show the first bound note that by \eqref{eq:epsrecurrence} and $\eps_{\ell}^{\max} \ge 0$ we have
\begin{align*}
    \eps^{\min}_{\ell+1} &\ge 1 - \prod_{i=1}^{\Delta} \left(1 + \frac{1}{C-i} \eps^{\max}_{\ell}\right) \ge 1 - \prod_{i=1}^{\Delta} \exp\left(\frac{1}{C-i}\eps^{\max}_{\ell}\right) = 1 - \exp\left(\eps_{\ell}^{\max} \sum_{i=1}^\Delta \frac{1}{C-i}\right) \\
    &\ge 1 - \exp\left(\eps_{\ell}^{\max} \left( \log\left(\frac{C}{C-\Delta}\right) + \frac{5}{C}\right) \right) \ge 1 - \left(1 + \frac{10}{C}\right)\exp\left(\eps_{\ell}^{\max} \log\left(\frac{C}{C-\Delta}\right)\right),
\end{align*}
where at the end we have used that $\eps^{\max}_\ell \le 1$ and $\exp(5/C) \le 1 + 10/C$. For the second claim we start with the inequality $1+x \ge \exp(x-x^2)$ for $x \ge -1/2.$ Now using \eqref{eq:epsrecurrence} and $\eps_{\ell}^{\min} \le 0$ gives us
\begin{align*}
    \eps^{\max}_{\ell+1} &\le 1 - \prod_{i=1}^\Delta \left(1 + \frac{1}{C-i}\eps^{\min}_{\ell} \right) \le 1 - \prod_{i=1}^\Delta \exp\left(\frac{1}{C-i}\eps^{\min}_{\ell} - \frac{(\eps^{\min}_{\ell})^2}{(C-i)^2} \right) \\
    &\le 1 - \exp\left(\eps_\ell^{\min}\sum_{i=1}^\Delta \frac{1}{C-i} - \sum_{i=1}^\Delta \frac{4}{(C-i)^2} \right) \\
    &\le 1 - \exp\left(\eps_\ell^{\min} \log\left(\frac{C}{C-\Delta}\right) - \frac{50}{C} \right) \le 1 - \left(1 - \frac{100}{C}\right)\exp\left( \log\left(\frac{C}{C-\Delta}\right)\eps_\ell^{\min}\right).
\end{align*}
Here at the end we have used that $\eps^{\min}_\ell \ge -2$ and the approximation in \eqref{eq:riemann}.
\end{proof}

Now define the function $f_\delta(x) := 1 - \exp((1-\delta)x)$ for any $\delta \in [0, 1).$ Note that if $\log(C/(C-\Delta)) = 1-\delta,$ then the iteration bound in \cref{lemma:onelevel} is basically given by $f_\delta(\eps)$ up to a $O(1/C)$ additive term. The following is essentially an immediate consequence of iterating \cref{lemma:onelevel} twice.
\begin{lemma}[Two-step error bound]
\label{lemma:twolevel}
If $C$ satisfies $\log(C/(C-\Delta)) \le 1-\delta$ with $\delta\in(0,1/2)$ then 
\[ \eps_{\ell+2}^{\max} \le f_\delta(f_\delta(\eps^{\max}_\ell)) + \frac{10^3}{C}. \]
\end{lemma}
\begin{proof}
By \cref{lemma:onelevel} we can start by bounding
\begin{align} \eps_{\ell+2}^{\max} \le f_{\delta}(\eps_{\ell+1}^{\min}) + \frac{10^2}{C} \exp((1-\delta)\eps_{\ell+1}^{\min}) \le f_{\delta}(\eps_{\ell+1}^{\min}) + \frac{100}{C} \label{eq:uselater2} \end{align}
as $\eps_{\ell}^{\min} \le 0$. Similarly, by \cref{lemma:onelevel} we get that
\begin{align} \eps_{\ell+1}^{\min} \ge f_{\delta}(\eps_{\ell}^{\max}) - \frac{10^2}{C} \exp((1-\delta)\eps_{\ell}^{\max}) \ge f_{\delta}(\eps_{\ell}^{\max}) - \frac{300}{C}, \label{eq:uselater1} \end{align}
as $\eps_{\ell+1}^{\max} \le 1$.
Now, note that $f_{\delta}(x)$ is $3$-Lipschitz for $x \le 1$ (as $|f_{\delta}'(x)| = (1-\delta)\exp((1-\delta)x)$). Combining \eqref{eq:uselater2}, \eqref{eq:uselater1}, and finally $3$-Lipschitzness of $f_{\delta}(x)$ gives
\begin{align*}
\eps_{\ell+2}^{\max} &\le f_{\delta}(\eps_{\ell+1}^{\min}) + \frac{100}{C} \\
&\le f_{\delta}(f_{\delta}(\eps_{\ell}^{\max}) - 300/C) + 100/C \le f_{\delta}(f_{\delta}(\eps_{\ell}^{\max})) + 1000/C.
\end{align*}
\end{proof}
\begin{lemma}
\label{lemma:numerical}
For any $\delta \in [0, 1/2)$ and $\eps \ge 0$ we have that $f_\delta(f_\delta(\eps)) \le (1-\delta)\eps.$
\end{lemma}
\begin{proof}
We need to argue that \[ 1 - \exp((1-\delta)(1-\exp((1-\delta)\eps))) \le (1-\delta)\eps. \]
This can be rearranged as
\[ \log(1-(1-\delta)\eps) \le (1-\delta)(1 - \exp((1-\delta)\eps)). \]
Taking a Taylor expansion and negating both sides gives the equivalent inequality
\[ \sum_{i=1}^\infty \frac{(1-\delta)^i\eps^i}{i!} \ge \sum_{i=1}^\infty \frac{(1-\delta)^{i+1}\eps^i}{i!}. \]
This is true term by term for $\eps \ge 0$ as desired.
\end{proof}
We can now use this claim to show \cref{thm:matchingtreelike}.
\begin{proof}[Proof of \cref{thm:matchingtreelike}]
We claim that if $C > \left(\frac{e}{e-1}+\delta\right)\Delta$, then $\log(C/(C-\Delta)) \le 1-\delta$, as required by \cref{lemma:twolevel}.
This latter condition is equivalent to $C \ge \frac{\exp(1-\delta)}{\exp(1-\delta)-1}\Delta.$ For $\delta < 1/20$, we have
\begin{align*} \frac{\exp(1-\delta)}{\exp(1-\delta)-1} &= 1 + \frac{1}{\exp(1-\delta)-1} \le 1 + \frac{1}{e(1-\delta)-1} \\
&= \frac{e}{e-1} + \frac{e\delta}{(e-1)(e-1-e\delta)} \\
&\le \frac{e}{e-1} + \delta \cdot \frac{e}{(e-1)(e-1-e/20)} < \frac{e}{e-1} + \delta.
\end{align*}
By \cref{lemma:obliviouschoice} we know that we can lower bound the probability that edge $e$ is matched by the probability edge $e$ is matched in some edge matching game (\cref{def:game}) where edges not in the witness tree are ignored, and the boundary edges are fixed obliviously beforehand.

Let us now bound the probabilities that $u, v$ are matched to a vertex below them. Note that these events are independent. Recall that we know that $\eps_u \le \eps^{\max}_g$ by definition.
We prove on induction on $\ell = 0, ..., g$ that
\[
\eps^{\max}_{2 \ell} \le \left( 1 - \frac{10^3}{\delta C} \right) (1 - \delta)^{\ell} + \frac{10^3}{\delta C} .
\]
For $\ell=0$ (base case) we have $\eps^{\max}_0 \le 1$.
Next, by \cref{lemma:twolevel,lemma:numerical} we know that
\begin{align*}
    \eps^{\max}_{2\ell+2} &\le f_{\delta}(f_{\delta}(\eps^{\max}_{2\ell})) + 10^3/C \\
    &\le (1-\delta)\eps^{\max}_{2\ell} + 10^3/C \\
    &\le \left( 1 - \frac{10^3}{\delta C} \right) (1 - \delta)^{\ell+1} + (1-\delta) \frac{10^3}{\delta C} + \frac{10^3 \delta}{\delta C} .
\end{align*}
For $\ell = g' \ge (g-1)/2$ we get
\begin{align} \max(\eps_u, \eps_v) \le \eps^{\max}_{2g'} \le \left(1 - \frac{10^3}{\delta C}\right)(1-\delta)^{(g-1)/2} + \frac{10^3}{\delta C} \le (1-\delta)^{(g-1)/2} + \frac{1000}{\delta C} \label{eq:epsbound} \end{align}

By \eqref{eq:epsbound}, the probability that edge $e$ is matched is at least
\begin{align*}
    q_u \cdot q_v \cdot \frac{C}{(C-c_u)(C-c_v)} 
    &= \frac{C-c_u}{C}(1-\eps_u) \cdot \frac{C-c_v}{C}(1-\eps_v) \frac{C}{(C-c_u)(C-c_v)} \\
    &\ge \frac{1}{C}\left(1 - (1-\delta)^{(g-1)/2} - \frac{1000}{\delta C}\right)^2.\qedhere
\end{align*}
\end{proof}

Given this it is essentially immediate to deduce our main result.
\begin{proof}[{Proof of \cref{thm:main}}]
To use the reduction of \cref{lemma:matching_reduction},
we will give an algorithm for online matching that,
given a graph of maximum degree $\Delta$,
matches every edge with probability at least $1/(\alpha \Delta)$,
with $\alpha = \frac{e}{e-1} + 3 \delta$
where $\delta$ will be determined at the end.

Our online matching algorithm is the following.
We first run \cref{algo:subsample} for some choice $\Delta' < \Delta$,
obtaining a graph $G'$,
and then run \cref{algo:matching} on $G'$, for the choice $C = \left(\frac{e}{e-1} + \delta\right)\Delta'$. Our goal is to pick $\delta$ as small as possible so that every edge $e$ is matched with
the required probability.

We will say that an edge $e \in E(G)$ is good if $e \in E(G')$ and its $g$-neighborhood in $G'$ contains no cycles ($g$ will also be determined below).
By \cref{lemma:degree_reduction},
the probability (with respect to randomness in \cref{algo:subsample}) that $e$ is good is at least
$\left( 1 - 5 \sqrt{\frac{\log \Delta'}{\Delta'}} - 3 \frac{(\Delta')^{5g}}{\Delta} \right) \frac{\Delta'}{\Delta}$.
By \cref{thm:matchingtreelike},
if $e$ is good, then it is matched
with probability (with respect to randomness in \cref{algo:matching}) at least
$\frac{1}{\left(\frac{e}{e-1} + \delta\right)\Delta'}\left(1 - (1-\delta/4)^{(g-1)/2} - \frac{10^4}{\delta C}\right)^2$.

To ensure that the product of these two probabilities is large enough,
we will choose $\delta, \Delta', g$ to attain the following bounds:
\begin{equation} \label{eq:bounds}
\max\left(5\sqrt{\frac{\log \Delta'}{\Delta'}}, \frac{3(\Delta')^{5g}}{\Delta}, (1-\delta/4)^\frac{g-1}{2}, \frac{10^4}{\delta C} \right) \le \frac{\delta}{10} .
\end{equation}
Then $e$ is matched with probability at least
\[
\left( 1 - \frac{\delta}{10} - \frac{\delta}{10} \right) \frac{\Delta'}{\Delta} \cdot \frac{1}{\left(\frac{e}{e-1} + \delta\right)\Delta'}\left(1 - \frac{\delta}{10} - \frac{\delta}{10} \right)^2
\ge
\frac{ 1 - \frac{3 \delta}{5}}{\left(\frac{e}{e-1} + \delta\right)\Delta}
\ge
\frac{ 1 }{\left(\frac{e}{e-1} + 3 \delta\right)\Delta} \,,
\]
where the last inequality is true for $\delta < \frac{1}{2}$.
Finally, the reduction of \cref{lemma:matching_reduction}
yields an edge coloring algorithm that,
for a graph of maximum degree $\Delta$,
outputs an
$\left(\frac{e}{e-1} + 3 \delta + O((\log n/\Delta)^{1/4})\right) \Delta$
edge coloring with high probability.

It remains to set $\delta, \Delta', g$ so as to satisfy \eqref{eq:bounds}.
First we set $\Delta' = 10^6 \delta^{-3}.$ This implies that $5\sqrt{\frac{\log \Delta'}{\Delta'}} \le \delta/10$, and $\frac{10^4}{\delta C} \le \frac{10^4}{\delta \Delta'} \le \delta/10$, as $C \ge \Delta'$. Additionally we set $g = \frac{100}{\delta}\log(100/\delta)$, so that
\[ (1-\delta/4)^{(g-1)/2} \le (1-\delta/4)^{\frac{40}{\delta}\log(100/\delta)} \le \exp(-\log(100/\delta)) \le \delta/10. \]
The final condition to check is that $3(\Delta')^{5g}/\Delta \le \delta/10.$ This is equivalent to
\[ 3(10^6\delta^{-3})^{\frac{500}{\delta}\log(100/\delta)} \le \Delta, \] so we may set $\delta = A\frac{(\log \log \Delta)^2}{\log \Delta}$ for some sufficiently large constant $A$.
\end{proof}

\subsection{Tightness of our analysis}
\label{subsec:tight}
We end this section with an informal discussion of why we believe that the analysis of our algorithm is tight barring major changes in the analysis style. In particular, we argue why there is a degree-$\Delta'$ graph and an edge $e$ such that running our algorithm on it will cause $e$ to be matched with probability $(1-\Omega(1))/C$ if $C < e/(e-1)\Delta'.$ Thus, if the analysis of our algorithm is improvable, then we must leverage additional properties of the subsampling procedure beyond the $\Delta'$ maximum degree bound and the fact that the neighborhood of each edge $e$ is a tree.

Consider a graph of maximum degree $\Delta'$ which is a tree rooted at an edge $e$, except that at the leaves, we add a gadget which is not a tree (i.e., it has a somewhat short cycle) and is a boundary edge for the tree rooted at $e$. If the algorithm processes the cycle edges first, then the ``errors'' $\eps_c$ for boundary edges $c$ in the cycles will all be $+\eps$. By propagating the errors upwards via the formula \eqref{eq:upwards}, the errors at higher levels will converge towards the order-two fixed point of $f(f(x))$ (for $f(x)$ defined above). Thus, $e$ will be sampled with probability $(1-\Omega(1))/C$.

\section{Conclusions, Open Problems, and Possible Improvements}
\label{sec:conclusions}
In this paper we have presented the first nontrivial algorithm for edge-coloring graphs in the online edge arrival setting against oblivious adversaries. In particular, we prove that one can use $(e/(e-1)+o(1))\Delta$ colors with high probability. We believe that the conjecture of Bar-Noy, Motwani, and Naor \cite{BMN92} is true, and that there is an online algorithm which requires only $\Delta + O(\sqrt \Delta \log n)$ colors.
However, it is not clear to us whether a $(1+o(1))\Delta$-edge coloring can be achieved in the \emph{adaptive} adversary setting. A good starting point would be to show a lower bound against deterministic algorithms.

We end by briefly discussing a concrete strategy based on the methods of this paper to give a $(1+o(1))\Delta$-edge coloring algorithm in the online setting. First note that the reduction to locally treelike graphs in fact shows that it suffices to consider the following approximate version of the problem: given a graph $G$ with maximum degree $\Delta$ and arbitrarily large girth (in terms of $\Delta$), can one partially edge-color $G$ with $\Delta$ colors such that every edge is colored with probability $1-o_{\Delta}(1)$? Note that a greedy algorithm gives a probability bound of $1/2-o_{\Delta}(1)$, while our matching-based algorithm can give $(e-1)/e-o_{\Delta}(1)$. 

As discussed in \cref{subsec:tight}, we believe that our current algorithm is the limit of ``local'' algorithms with only two states (matched or unmatched) and therefore the key issue is analyzing ``local'' algorithms with more than two states. For concreteness, we believe the following algorithm colors well on trees of maximum degree $\Delta$: with probability $\epsilon$ leave an edge blank and with probability $1-\epsilon$ color it with a random color among the set of available colors (leaving the edge blank in the case there are no available colors). Note here that there are $\Delta$ colors and hence there are $\Delta+1$ states. This algorithm can hypothetically be analyzed in the framework given in this paper; in particular the reduction to understanding the corresponding ``adversarial'' game on trees remains unchanged. The key issue is in analyzing the ``adversarial'' game.

For analyzing the adversarial game, our analysis relies on monotonicity (\cref{lemma:obliviouschoice}) in order to reduce to a non-adaptive adversary which can be understood through tree recurrences. The key issue is that the algorithm described above is not obviously monotone in any parameter and hence understanding the probability that an edge is unmatched via tree recurrences is substantially more difficult. This lack of monotonicity is closely related to one reason why the algorithm of Weitz \cite{Wei06} for estimating the partition function of the hardcore model is not currently known to extend to the setting of non-ferromagnetic models (e.g.,~approximating the number of colorings of a $\Delta$-regular graph). Furthermore, note that using tree recurrences one can still analyze adversaries which are non-adaptive, and we believe that in this setting, given large girth, one could apply tree recurrences to analyze the algorithm above. However, such a result for ``non-adaptive'' adversaries, without monotonicity, does not imply anything regarding the original algorithm. 

On the positive side, one plausible reason why the analysis of an algorithm with more than two states may be tractable is due to the fact that we can assume that our graphs are high-girth. This could potentially parallel the analysis of Glauber dynamics for vertex-coloring graphs with maximum degree $\Delta$ using $(1+o(1))\Delta$ colors. While such a result on the mixing time of Glauber dynamics on graphs with large girth (in terms of the maximum degree) is not known in the literature, several results in this direction are known \cite{HV03,GMV05,HV06,FGYZ21}. Hence we suspect that a result of this form may be within reach via spectral or coupling-based techniques \cite{anari2021spectral} for analyzing Markov chains. However, it is currently unclear to us how to apply such spectral or coupling-based techniques to analyze the random process arising in the edge-coloring setting.

\section{Acknowledgements}
Janardhan Kulkarni would like to thank Sayan Bhattacharya, Fabrizio Grandoni, and David Wajc for initial discussions on this problem.

Yang P. Liu was partially supported by by the Department of Defense (DoD) through the National Defense Science and Engineering Graduate Fellowship. Sah and Sawhney were supported by NSF Graduate Research Fellowship Program DGE-1745302. Part of this work was done while Liu and Sawhney were interns at Microsoft Research, Redmond.

\bibliographystyle{alpha}
\bibliography{main.bib}

\appendix

\section{Random Edge Order Arrival}
\label{sec:randomorder}
\subsection{Sparsification}
We now introduce the sparsification procedure used in the random order case. The proof here is more delicate than \cref{algo:subsample} as we need to preserve various properties of being a randomly ordered subset in the subsampling algorithm. 

\begin{algorithm}[!ht]
\caption{Divides the graph $G$ of maximum degree at most $\Delta$ into a series of maximum degree approximately $\Delta'$ graphs. \label{algo:subsample-random}}
\SetKwProg{Proc}{procedure}{}{}
\Proc{$\textsc{Split}(G, \Delta',\Delta)$}{
    $T\assign \lceil \Delta/\Delta'\rceil$\\
    $E \assign 3\sqrt{\Delta'\log \Delta'}$\\
    $G_1,G_2,\ldots,G_{T},R \assign \emptyset.$ \Comment{Initially empty subgraphs.} \\
    $m_v^{1},\ldots, m_v^{T}\assign 0$ for $v \in V(G).$ \Comment{Number of adjacent marked edges to vertex $v$ of each color} \\
    \For{$i = 1, \dots, m$}{
        \tcp{Edge $e_i = (u, v)$ arrives}
        $C \assign \tsc{unif}(\{1,\ldots,T\})$ \Comment{We say edge $e$ is marked color $C$.} \\
        $m_u^{C} \assign m_u + 1, m_v^{C} \assign m_v + 1.$\\
            \If{$m_u \le \Delta' + E$ \emph{and} $m_v \le \Delta'+ E$}{
                $E(G_C) \assign E(G_C) \cup \{ e \}$ \\
            }
            \Else{
            $E(R) \assign E(R) \cup \{ e \}$
            }
    }
}
\end{algorithm}

\begin{lemma}\label{lem:random-degree-reduction}
Let $\Delta'\ge 2$. Running \cref{algo:subsample-random} on graph $G$ with $n$ vertices and maximum degree at most $\Delta$ outputs a uniformly random partition of the edges of $G$ into $T$ subgraphs $G_1', \dots, G_T'$ (i.e.~edge $e$ is in each $G_i'$ with probability exactly $1/T$), and subgraphs $G_i \subseteq G_i'$ for $i \in [T]$ satisfying the following properties.
\begin{itemize}
    \item The maximum degree in $G_i$ is $\Delta'+3\sqrt{\Delta' \log \Delta'}.$
    \item With probability at least $1-1/n$ we have that for all $v \in V(G)$ that \[ \deg_G(v) - \sum_{i \in [T]} \deg_{G_i}(v) \le  C_{\ref{lem:random-degree-reduction}}(\Delta/(\Delta')^2 + \sqrt{\Delta \log n}). \]
\end{itemize}
\end{lemma}
\begin{proof}
The first bullet point follows directly by the algorithm.

The second bullet point is slightly more intricate. We consider a fixed vertex $v$ and note that for an edge $e = (u,v)$ of color $i$ to be deleted we have that either the endpoint $u$ or $v$ has more than $\Delta' + 3\sqrt{\Delta'\log \Delta'}$ edges marked with color $i$. We say that a vertex $v$ is bad with respect to the color $i$ if at least $\Delta' + 3\sqrt{\Delta'\log \Delta'}$ edges have been marked with the color $i$. Fix the ordering of edges $e_1,\ldots,e_m$ and let $\mc{E}_{e}$ denote the event that there are at least $\ell$ bad colors in the graph for some vertex in the graph $G$ when $e$ is processed. Note $\mc{E}_{e_i}\subseteq\mc{E}_{e_j}$ for $i\le j$ since the set of processed edges grows. Let $Y_v$ denote the number of deleted edges at the vertex $v$ and suppose the edges of $v$ are presented as $(v,u_1),\ldots,(v,u_{\deg{v}})$. Let $X_v^{u_i} = Y_v\mbm{1}_{\mc{E}_{(v,u_i)}^{c}}$ and note that $X_v^{u_{i+1}}\le X_v^{u_i}+1$ and that $\mb{P}[X_v^{u_{i+1}} = X_v^{u_i}+1] \le (2\ell)/\Delta$ as there are at most $\ell$ bad colors at $u_i$ and $v$ before $(v,u_i)$ is processed if $\mbm{1}_{\mc{E}_{(v,u_i)}^{c}}$ holds. Therefore we have $\mb{P}[X_v^{u_{\deg(v)}}\ge 2\ell + \sqrt{\Delta}t]\le \exp(-\Omega(t^2))$ by stochastic domination and Chernoff.

We now bound the event $\mc{E}_{e}$ for the final edge $e$ in the ordering; note that if $\mc{E}_{e}^c$ holds then $X_v^{u_{\deg(v)}}$ serves as an upper bound for the number of deleted edges at the vertex $v$. Let $Z_i$ be the number of edges emanating from $v$ marked $i$.  Note that $(Z_i)_{i\in[T]}$ are negatively associated (see e.g.~\cite[Proposition~E.6]{BGW21}) and if $Y_i = \mbm{1}_{Z_i\ge\Delta' + 3\sqrt{\Delta'\log(\Delta')}}$ by Chernoff $\mb{P}[Y_i = 1]\le 1/(\Delta')^2$. Hence the expected number of bad colors at a vertex is bounded by $\Delta/(\Delta')^2$, and thus by Hoeffding's inequality for negatively associated random variables we see that $\mb{P}[\sum_{i\in T}Y_i\ge \Delta/(\Delta')^2 + \sqrt{T}t]\le\exp(-\Omega(t^2))$. Setting $t$ to be a sufficiently large multiple of $\sqrt{\log n}$, we see that with probability at most $1/n^2$ no vertex has more than $\Delta/(\Delta')^2 + C\sqrt{T\log n}$ bad colors. Setting $\ell$ as such and $t$ in the previous bound we find that with probability $1-1/n$ the omitted edges at every vertex is less than $C(\Delta/(\Delta')^2 + \sqrt{\Delta \log n})$.
\end{proof}

\subsection{Tree recurrence calculation}\label{sec:tree-1}

We first describe a simple algorithm for sampling a matching on a graph. We then argue that each edge is matched with the desired probability on any tree $G$. Fix a sampling parameter $C$. When an edge $e = (u,v)$ arrives, if it has no neighbors in the matching we add it in with probability $\frac{C}{(C-d_u)(C-d_v)}$, where $d_u$ and $d_v$ are the degrees of $u$ and $v$ in the edges that have already arrived. Note that if $C > \Delta(G) + 2\sqrt{\Delta(G)}$ the the sampling probability is in $[0, 1]$ so it is well-defined. This is \cref{algo:matching}.

\begin{lemma}\label{lem:tree-exact-probability}
If one runs \cref{algo:matching} on a tree, then for every $e\in E(G)$, it is included in the matching with probability exactly $1/C$.
\end{lemma}
\begin{proof}
We prove this via induction on the position of $e$. Clearly if it is the first edge, it is in with probability $C/C^2 = 1/C$. Now suppose it appears later. In order for $e = (u,v)$ to be in the matching, none of the adjoining prior edges may be chosen. Suppose they are $(u,u_1),\ldots,(u,u_{d_u})$ and $(v,v_1),\ldots,(v,v_{d_v})$. By induction, each $(u,u_i)$ appears with probability $1/\Delta$. These are clearly disjoint events by definition, so with probability $1-d_u/C$ no such edge has been included yet. Similarly, with probability $1-d_v/C$ no edge connected to $v$ has been included yet. These two events are clearly independent (as the corresponding graphs are disconnected), so the probability $e$ is included is easily seen to equal
\[\Big(1-\frac{d_u}{C}\Big)\Big(1-\frac{d_v}{C}\Big)\cdot\frac{C}{(C-d_u)(C-d_v)} = \frac{1}{C}.\qedhere\]
\end{proof}

We next use this to describe a coloring algorithm on trees which provides the necessary guarantees.

\begin{algorithm}[!ht]
\caption{Computes an approximate coloring on a graph $G$ with a fixed order on the edges $e_1,\ldots,e_m$ and a maximum degree at most $\Delta$. \label{algo:random-coloring}}
\SetKwProg{Proc}{procedure}{}{}
\Proc{$\textsc{Tree-Coloring}(G)$}{
    $G\assign \{e_1,\ldots,e_m\}$\\
    $R \assign \{\}$\\
    $C = \Delta + \Delta^{3/4}$
    \For{$i = 1, \dots, \Delta$}{
        Remove all vertices in $G$ of degree at least $C$ and add to $R$ (with edges) \label{line:remove2} \\
        Run $\textsc{Matching}(G, C)$ and let the output matching be $\mc{M}_i$\\
        $G\assign G\setminus \mc{M}_i$\\
        $C\assign C-1+\Delta^{-1/12}$\\
    }
    Output $G\cup R$ \Comment{Non-colored edges}
}
\end{algorithm}

As written the above algorithm is sequential, however we will prove that the above algorithm can be simulated online and that it provides the necessary coloring guarantee.

\begin{lemma}\label{lem:random-color-tree}
If one runs \cref{algo:random-coloring} with $\Delta\ge\Delta_{\ref{lem:random-color-tree}}$ on a tree of maximum degree at most $\Delta$, then for every $e\in E(G)$ it is output as a non-colored edge with probability at most $\Delta^{-1/24}$. Additionally, \cref{algo:random-coloring} can be implemented online for any graph.
\end{lemma}
\begin{proof}
First, we clearly see that \cref{algo:random-coloring} can be implemented online: simply run the stages of edge coloring in parallel, passing to the next stage as needed. Since the instances of \cref{algo:matching} which are called in the iteration are themselves online algorithms, one can maintain the necessary data in parallel online to run them.

Recall $C = \Delta+\Delta^{3/4}$. We initially show that for a fixed vertex $v$, its probability of removal is low. Notice that at every stage $i$ of the matching it has a probability $d_v/(C-(i-1)(1-\Delta^{-1/12}))$ of its degree $d_v$ decreasing by $1$. It is removed if $d_v > \Delta+\Delta^{3/4}-(i-1)(1-\Delta^{-1/12})$ occurs.

The point is that if ever $d_v\ge\Delta-(i-1)(1-\Delta^{-1/8})$, then $d_v$ stays the same with probability at most $2\Delta^{3/4-11/12} = 2\Delta^{-1/6}$. If $v$ is removed, consider the last point in time that $d_v < \Delta-(i-1)(1-\Delta^{-1/12})$. After this point and until it is removed, we have a consecutive string of runs in which the quantity $d_v+(i-1)(1-\Delta^{-1/12})$ must go up by a total of $\Delta^{3/4}$, yet it goes up by $1-\Delta^{-1/12}$ with probability at most $2\Delta^{-1/6}$ and goes down by $-\Delta^{-1/12}$ with probability at least $1-2\Delta^{-1/6}$ (those being the only two options). This step clearly has a \emph{negative} mean with magnitude at least $\Omega(\Delta^{-1/12})$. We also see that it will take at least $\Delta^{3/4}$ steps and at most $\Delta$ steps. Therefore, the probability that it gets removed after $t$ steps is bounded by $\exp(-\Omega(\Delta^{3/4}))$ by the Chernoff bound. Taking a union bound over the possible values of $t$, we see that the probability that $v$ is removed is bounded by $\exp(-\Omega(\Delta^{3/4}))$.

Next, for a fixed edge $e$, we therefore see that either of its two vertices are removed with probability at most $\exp(-\Omega(\Delta^{3/4}))$. In the remaining cases, at each stage by \cref{lem:tree-exact-probability} it is included with probability $1/(C-(i-1)(1-\Delta^{-1/12}))$, conditional on the previous outcomes. Therefore its probability of not being included in this case is
\begin{align*}
\prod_{i=1}^\Delta\bigg(1-\frac{1}{C-(i-1)(1-\Delta^{-1/12})}\bigg)&\le\exp\bigg(-\sum_{i=1}^\Delta\frac{1}{C-(i-1)(1-\Delta^{-1/12})}\bigg)\\
&\le\exp\bigg(-\frac{1}{1-\Delta^{-1/12}}\int_{C-(\Delta-1)(1-\Delta^{-1/12})}^C\frac{dt}{t}\bigg)\le\Delta^{-1/24}
\end{align*}
for $\Delta$ sufficiently large.
\end{proof}

We now prove that in the random order case it suffices to understand the tree case. The intuition is that for an edge $e$ and an arrival ordering on other edges, an edge $f$ does not affect the coloring of edge $e$ if there is no path from $f$ to $e$ in the order given by the input.
\begin{definition}[Witness edges]
\label{def:witnessedge}
Consider a graph $G$ with $m$ edges and a permutation $e_1, \dots, e_m$ on the edges of $G$. For an edge $e = e_{i(e)}$, and an edge $e_i \neq e$ we say that $e_i$ is a \emph{witness} for $e$ if there are indices $i = i_0 < i_1 < \dots < i_\ell = i(e)$ such that the edges $e_{i_0}, e_{i_1}, \dots, e_{i_\ell}$ for a path.
\end{definition}
Note that $e$ is a witness of itself.
The key point is that the probability an edge $e$ is matched during an invocation of \cref{algo:random-coloring} only depends on the witness edges (and their ordering).
\begin{lemma}
\label{lemma:witnessonly}
Given an ordering of edges in a graph $G$ and an edge $e$, let $W$ be the set of witness edges (\cref{def:witnessedge}). Then the probability that edge $e$ is matched by \cref{algo:random-coloring} when run on graph $G$ is the same as when run as only edges in $W$ with the same ordering.
\end{lemma}
\begin{proof}
We use induction of suffixes of edges in the ordering. Let us consider appending a single edge $f$ to the start of an ordering. This does not affect whether later edges are witnesses. Thus it suffices to argue that if $f$ affects the probability that $e$ is matched by \cref{algo:random-coloring}, then it is a witness. Indeed, this only happens (by induction) if $f$ is adjacent to some later witness edge $f'$ because the sampling probabilities in \cref{algo:matching} only depend on current degrees of the vertices adjacent to an edge. By induction, there is a path from $f'$ to edge $e$ (by the definition of witness edge), so there is a path from $f$ to $e$ as desired.
\end{proof}
Surprisingly, for a random ordering, almost every edge in the sampled subgraph returned by \cref{algo:subsample-random} has that its witnesses form a tree.
\begin{lemma}\label{lem:tree-cuts}
Consider a graph $G$, and parameters $\Delta', g$ with $g \ge 10\Delta'.$ The probability that for all vertices $v$ that the number of subgraphs $G_i$ returned by \cref{algo:subsample-random} where some edge $e$ adjacent to $v$ has witnesses that \emph{do not} form a tree is at most
$4T(2e\Delta'/g)^{g} + 6(\Delta')^{5g} + C_{\ref{lem:tree-cuts}}\sqrt{T \log n}$ is at least $1/n$ over random edge orderings for some absolute constant $C_{\ref{lem:tree-cuts}}$.
\end{lemma}
\begin{proof}
We show the desired claim instead for the random subgraphs $G_i'$ -- this suffices as $G_i$ is a subgraph of $G_i'.$
Define $E_i(v)$ as the event (over randomness used to generate $G_i'$ and the random ordering) where vertex $v$ in $G_i'$ has a cycle in its $g$-neighborhood, or some witness edges for $v$ are distance at least $g$ from $v$. Note that $\neg E_i(v)$ implies that the witness edges of $v$ in $G_i'$ form a tree.

We upper bound $\mb{P}[E_i(v)].$ We start by bounding the probability that $v$ has a cycle in This is at most $3(\Delta')^{5g}/\Delta$ by \cref{lemma:highgirth}.
Now we bound the probability that an edge of distance more than $g$ from $v$ in $G_i'$ is a witness. The number of paths of length $\ell$ starting at $v$ is bounded by $\Delta^\ell.$ The probability that such a path has an endpoint edge which is a witness is at most $1/\ell!$, as the edges in the random order must be in the order of the path. Hence the probability that some edge of distance more than $g$ from $v$ is a witness in $G_i'$ is at most
\[ \sum_{\ell \ge g} \Delta^\ell/\ell! \cdot (\Delta'/\Delta)^\ell \le \sum_{\ell \ge g} (e\Delta'/\ell)^\ell \le 2(2e\Delta'/g)^{g} \]
for $g \ge 10\Delta'.$ Hence $\mb{P}[E_i(v)] \le 3(\Delta')^{5g}/\Delta + 2(2e\Delta'/g)^{g}$ for any $v$.

Note that the variables $E_i(v)$ for $i \in [T]$ are negatively associated because they are all monotone graph properties (see \cite[Appendix~E]{BGW21}) and hence by Hoeffding for negatively associated random variables (see e.g.~\cite[Proposition~E.6]{BGW21}) the result follows.

\end{proof}

\subsection{Completing the Proof}
We now complete the proof of the \cref{thm:random}. In order to prove the desired result consider the following algorithm. 
\begin{itemize}
    \item Divide $G$ in $G_1,\ldots,G_T,R$ using \cref{algo:subsample-random}.
    \item Color $G_1,\ldots,G_T$ using \cref{algo:random-coloring}.
    \item Color remaining edges in $R$ and not colored in the previous step using the greedy algorithm.
\end{itemize}

\begin{proof}[Proof of \cref{thm:random}]
In order to prove \cref{thm:random} it suffices to prove that the above algorithm can be simulated online and that the union of the remainder graphs coming from $G_1,\ldots,G_T$ and $R$ when combined have maximum degree $o(\Delta)$. Furthermore it will be useful to define the graph $R'$ consisting of adjacent edges to vertices $v$ in some $G_i'$ whose witness edges (\cref{def:witnessedge}) do not form a tree. The first claim is clear as both \cref{algo:subsample-random,algo:random-coloring} can be implemented online. 

By the second bullet point of \cref{lem:random-degree-reduction} we have with probability $1-1/n$ that the maximum degree of $R$ is at most $C_{\ref{lem:random-degree-reduction}}(\Delta/(\Delta')^2 + \sqrt{\Delta \log n})$. The additional edges in $R'$ come from edges adjacent to vertices $v$ in some $G_i'$ whose witness edges (\cref{def:witnessedge}) do not form a tree. The maximum degree in each $G_i$ is at most $2\Delta'$ by the first bullet point of \cref{lem:random-degree-reduction}. Hence by \cref{lem:tree-cuts} with probability at least $1-1/n$ the degree of any vertex in $R'$ is bounded by
\begin{align} &C_{\ref{lem:random-degree-reduction}}(\Delta/(\Delta')^2 + \sqrt{\Delta \log n}) + 2\Delta' \cdot \left(2T(2e\Delta'/g)^{g/2} + (\Delta')^g + C_{\ref{lem:tree-cuts}}\sqrt{T \log n} \right) \nonumber \\
\le ~& O\left(\frac{\Delta}{(\Delta')^2} + \Delta(2e\Delta'/g)^{g/2} + (\Delta')^g + \sqrt{\Delta\Delta'\log n} \right) \label{eq:degrprime}
\end{align}
as long as $g \ge 10\Delta'.$ Let the expression in \eqref{eq:degrprime} be $\Delta_{R'}.$

Now we consider additional edges removed from the union of $G_1, \dots, G_i$ through line \ref{line:remove2} of \cref{algo:random-coloring}. Note that \cref{algo:random-coloring} runs independently for $G_1,\ldots,G_T$. By the linearity of expectation, for a fixed vertex the expected number of remaining edges not in $R'$ and not chosen is at most $3(\Delta')^{23/24}$. Furthermore this is at most $2\Delta'$. Hence by Azuma-Hoeffding the probability that there are more than $6(\Delta')^{23/24}T + C\sqrt{T\log n}\Delta' = 6\Delta/(\Delta')^{1/24} + C\sqrt{\Delta\Delta'\log n}$ at a particular vertex is at most $1/n^2$ for $C$ sufficiently large. Note that here we have implicitly assumed that $g\ge 2\Delta'$.

Therefore with probability at least $1-3/n$ we have that the maximum degree in the leftover graph is at most
\[ \Delta_{R'} + 6\Delta/(\Delta')^{1/24} + C\sqrt{\Delta\Delta'\log n}. \]
assuming that $g\ge 10\Delta'$. Now setting $g = 10\Delta'$, $\Delta' = c\min(\log\Delta/\log\log\Delta,\sqrt{\Delta/\log n})$ for a sufficiently small absolute constant $c > 0$ and noting that $\Delta = \omega(\log n)$ implies that $\Delta' = \omega(1)$ gives the desired result.
\end{proof}

\end{document}